\documentclass[12pt]{article}

\usepackage{sbc-template}

\usepackage{graphicx,url}
\usepackage[rgb]{xcolor}
\usepackage{geometry}
\usepackage{enumitem}

\usepackage[brazil]{babel}   
\usepackage[utf8]{inputenc}  

\newcommand{\eat}[1]{}

\usepackage{xspace}
\usepackage{amsmath}

\usepackage{amsthm}

\usepackage{amssymb}
\usepackage{amsfonts}
\usepackage{comment}
\usepackage{bm}
\usepackage{comment}
\usepackage{graphics}
\usepackage{graphicx}
\usepackage{listings}
\usepackage{url}
\usepackage{lscape}
\usepackage{mathrsfs}

\usepackage{graphicx}
\usepackage{latexsym}

\usepackage{xspace}
\usepackage{dsfont}
\usepackage{amsmath}
\usepackage{amssymb}
\usepackage{amsfonts}
\usepackage{comment}
\usepackage{bm}
\usepackage{comment}
\usepackage{graphics}
\usepackage{graphicx}
\usepackage{listings}
\usepackage{url}
\usepackage{lscape}
\usepackage{mathrsfs}

\usepackage{cite}
\usepackage{graphicx}
\usepackage{verbatim}
\usepackage{url}
\usepackage{mathrsfs}
\usepackage{amsmath}

\usepackage{graphicx} 
\newcommand {\beq}{\begin{equation}}
\newcommand {\eeq}{\end{equation}}
\newcommand {\barr}{\begin{array}}
\newcommand {\earr}{\end{array}}
\newcommand {\bearn}{\begin{eqnarray*}}
\newcommand {\eearn}{\end{eqnarray*}}
\newcommand {\bear}{\begin{eqnarray}}
\newcommand {\eear}{\end{eqnarray}}

\newtheorem{theorem}{Theorem}[section]

\newtheorem{proposition}{Proposition}[section]

\newtheorem{definition}{Definition}[section]
\renewenvironment{definition}{\refstepcounter{definition}\par\noindent\textit{Definition~\thedefinition:}\xspace}{\nobreak\hfill$\Diamond$\par}
\newtheorem{example}{Example}[section]
\newlength{\labelexample}
\title{Flexible Content Placement in Cache Networks using Reinforced Counters}

\author{Guilherme Domingues\inst{1}, \\ Edmundo de Souza e Silva\inst{1}, Rosa M. M. Leao\inst{1}, Daniel S. Menasche\inst{1}}

\address{Universidade Federal do Rio de Janeiro
  (UFRJ)\\
  Rio de Janeiro, RJ -- Brazil
  \email{\{guilhdom, edmundo, rosam, sadoc\}@land.ufrj.br}
}

\begin{document} 

\maketitle


\begin{abstract}
In this paper we study the problem of content placement  in a cache network.  We consider
a network where routing of requests  is based on random walks. Content placement is done using
a novel 
mechanism referred to as ``reinforced counters''.   To each content we associate a counter,
which is incremented every time the content is requested, and which is decremented at a 
fixed rate.   We model and analyze this mechanism, tuning its parameters so as to achieve desired 
performance goals for a single cache or for a cache 
network.  We also show that the optimal static content placement, without reinforced counters, is
NP hard under different design goals.  
\end{abstract}


\section{Introduction}

In today's Internet the demand for multimedia files and the sizes of these files are 
steadily increasing. 
The popularity of Youtube, Dropbox and one-click
file sharing systems, such as Megaupload, 
motivate researchers to seek for novel cost-effective content 
dissemination solutions. 

\emph{Caching} is one of the most classical solutions to increase the load supported by computer systems.  
In essence, caching consists of transparently storing data in a way that future requests can be served faster.
Caching in the realm of standalone computer architectures received significant attention from 
the research community since the early sixties, 
and web-caching has also been studied for at least one decade.  
The objects of study of this work, in contrast, are cache networks, which were proposed
and  started to receive focus much more recently~\cite{non, ccn1}.

Cache networks comprise two main features: caching and routing,
both performed by the same core component: a cache router.
In a cache network, content traverses the network from
sources to destinations, but can in turn be stored in caches strategically 
placed on top of the routers.  
The \emph{data plane} is responsible for transmitting contents whereas the \emph{control plane} 
is responsible for the routing of requests.  
The system considered in this paper is illustrated in Figure~\ref{system}.  
A request for file $F$ is routed using random walks through
the caches.  
Every time the request hits a cache, the cache selects uniformly at random one of the \eat{other} 
links that are incident to the cache
and uses it to forward the request.  
Once the request reaches a cache where the content is stored, the content is transferred to the requester
through the data plane.
We assume that all contents are stored in at least one cache.  
%
\begin{figure}
\begin{center}
\includegraphics[scale=0.20]{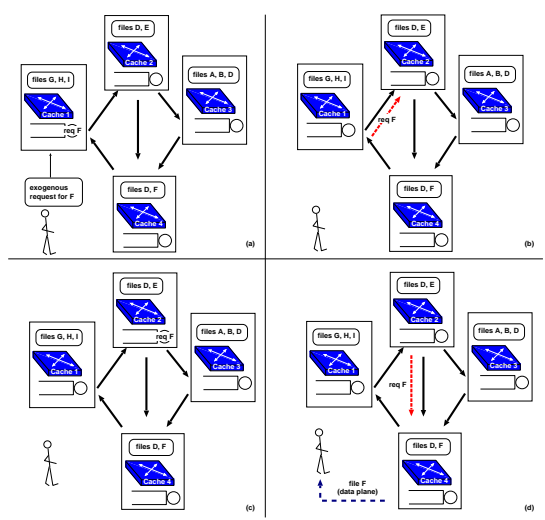} 
\end{center}
\caption{Illustration of how system works (cache-router symbol introduced in~\cite{breadcrumbs1}).  $(a)$ An exogenous request for file $F$ is placed and processed in cache 1;  $(b)$ the request is routed through the control plane, and yields an endogenous request  to cache 2; $(c)$ the request is processed in cache 2; $(d)$ the request is routed through the control plane to cache 4. The file is found and is transferred to the user through the data plane.}
\label{system}
\end{figure}


In this paper we study the problem of content placement in cache networks.   
We pose the following questions: (1) How to store and evict contents from a cache in a 
flexible and scalable manner? (2) How to efficiently and distributedly place content in a cache network?
%

Content placement is done using a novel 
mechanism referred to as ``reinforced counters'' \cite{sbrc13,wperf14}. 
To each content we associate a counter,
which is incremented every time the content is requested, and which is decremented at a 
fixed rate.   
We model and analyze this mechanism, tuning its parameters so as to achieve desired 
performance goals for a single cache or for a cache network. 

Each cache in a cache network has a given service capacity.  
The service capacity of a cache is related to the time it takes to 1) find
content in the cache, 2) return that content to the user, through the data plane,  
in case of a cache hit and 3) route the request to another cache, through the control plane, 
in case of a cache miss.  
Note that in the two latter cases, the service capacity of a cache network accounts for network delays due to 
transmission and queueing.
A cache is stable if, for a given workload, the queue of pending requests does not grow unboundedly with 
respect to time.
A cache network is stable if, for a given workload, all its caches are stable.
We  show that the optimal static content placement, without reinforced counters, is
NP hard under different design goals accounting for stability and content availability.  

To summarize, in (partially) answering the questions above we make the following contributions
\begin{itemize}
\item we introduce reinforced counters as a way to flexibly store and evict contents from a cache, showing that they are amenable to analytical study and optimal tuning  
\item we propose a new formulation of the optimal content placement in cache networks accounting for stability. Then, we show that static content placement, without reinforced counters, is NP hard, which motivates the use of reinforced counters or variants in a network setting.
\end{itemize}

The rest of this paper is organized as follows.  
Section \textsection\ref{sec:reinforced} studies the single cache police based on 
the reinforced counters.  
In \textsection\ref{model} we address the problem of multiple caches
and 
we show that under different setting the problem at hand is NP hard. 
In \textsection\ref{related} we present related work, and \textsection\ref{concl}
concludes.

\section{Single Cache: Reinforced Counters and Flexible Content Placement}
\label{sec:reinforced}

In this section, we study  placement and eviction policies for a single cache
under the assumption that
the dynamics of each content is decoupled from the others. 
Decoupled content dynamics
yields tractable analysis and can be used for approximate 
systems with fixed capacity.
They are also of interest in the context of DNS caches and the novel
Amazon ElastiCache system (\url{http://aws.amazon.com/elasticache/}).

In the policies to be introduced in this section, the expected
number of items in the system can be controlled. 
Let $\pi_{up}$ be the probability that each  content is stored in the cache. 
Given a collection of $N$ contents the expected number of stored contents is $N \pi_{up}$,
which can be controlled or bounded according to users needs.  
In what follows, we consider the problem of controlling $\pi_{up}$ using
a simple mechanism referred to as \emph{reinforced counters} (Section \ref{sec:single-threshold}). 
In Section \ref{sec:single-t-ex}
we illustrate how the mechanism works through some simple numerical examples.   
Then, in Section \ref{sec:hyst} we extend the reinforced counters to allow for hysteresis,
showing the benefits of hysteresis
for higher predictability  and reduced chances of content removal before full download.

\subsection{Reinforced counter with a single threshold}
\label{sec:single-threshold}

We consider a special class of content placement mechanisms, henceforth
referred to as reinforced counters
\cite{sbrc13,wperf14}. 
To each content we associate a reinforced counter,
which is increased by one  every time the content is requested, and is decreased by one 
as a timer ticks.  
Let $K$ be the reinforced counter eviction threshold.
Whenever the counter is decremented from $K+1$ to  $K$ the content is evicted.
The timer ticks every $1/\mu$ seconds.
Henceforth, we assume that the time between ticks
is exponentially distributed.   
This mechanism is similar to TTL caches, employed
by DNS and web-caching systems~\cite{exact, towsley1} with additional flexibility to allow
for content to be inserted in the cache only after the threshold $K$ is surpassed.  
The advantages of using a threshold $K$ over existing mechanisms will be shown in this section.

We assume that the behavior of each of the contents is decoupled.
This assumption is of interest  in  at least three scenarios. 
\begin{itemize}

\item {\em Cache capacity is infinite} - 
In the cloud, we may assume that storage is infinite. Users
might incur costs associated to larger amounts of stored content, but the storage space itself
is unlimited.  Therefore, costs are proportional to the space used, and constraints are soft (as opposed
to hard). 

\item  {\em Approximation of the behavior of fixed capacity system} - 
Under the mean field approximation~\cite{towsley1}, it has been
shown that decoupling the dynamics of multiple contents can lead  to reasonable approximations
to the content hit probabilities.    In this case, instead of considering that the cache has
a fixed capacity, it is assumed that there are constraints on the expected number of items
in the cache.
Note that under this approximation the cache capacity design problem can be compared to
the problem of calculating the capacity of communication lines in a telephone network.

\item {\em Time-to-live caches} - 
Roughly, when capacity is finite, one can take advantage of
statistical multiplexing either by
evicting items from the cache when (a) an overflow occurs or 
(b) when the item expire. 
In the latter case, if the probability of cache overflow is low,
the approximate analysis of  hit probabilities can be  done assuming that each item is
decoupled from the others~\cite{towsley1}.
This is particularly relevant in the context of
time-to-live caches (such as those used by the DNS system) where entries need to be
renewed from time to time to avoid staleness.

\end{itemize}

Our goal in the remainder of this section is to show the advantages of having a content
placement mechanism with two associated knobs, $\mu$ and $K$, allowing for the user
to fine tune both the fraction of time in which the content is in the cache (steady state metric) 
while at the same time controlling the mean time between content insertions 
(or, equivalently, controlling the rate at which content is evicted or brought back into cache)
so as to avoid that content is replaced too fast and content starvation (that is, content
is never included into cache or removed from it during a finite but large time interval).
The timer tick rate $\mu$ and the reinforced counter threshold $K$ can be tuned so as to adjust 
the long term fraction of time in which the content is stored in the cache
and to guarantee that the mean time between content eviction and content reinsertion into the cache is bounded.

In what follows, we assume
that requests for a given content arrive according to a Poisson process with rate $\lambda$.
$\lambda$ is also referred to as the  content popularity.
(Recall that the behavior of each content is decoupled for others.)
Figure \ref{fig:cache-in-out} is useful to illustrate the different intervals of the
cache content replacement and the notation we use.
The blue (red) intervals in the figure indicate that the content is stored (or not) in cache.
If content is not in cache it is brought into cache when a new request for it arrives
and the reinforced counter is at the threshold $K$.
On the other hand, if the value of the reinforced counter is at $K+1$ and
the counter ticks, the content is removed from cache.
Note that we adopt the same assumption as in~\cite{towsley1}:
the insertion and eviction of a content  is not influenced by other contents in the same cache.
As mentioned above, fixed storage in cache is modeled by considering the expected
number of contents into the cache.
%
\begin{figure}[h!]
\center
\includegraphics[scale=0.6]{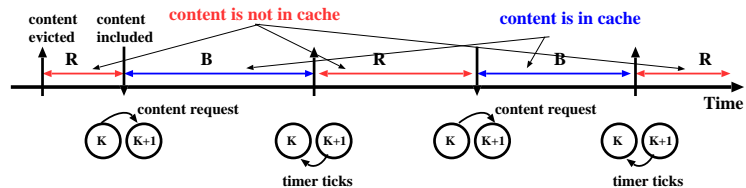}
\caption{Request counter and notation.
\label{fig:cache-in-out}
}
\end{figure}

Let $\pi_{up}$ be the fraction of
time in which the content is in the cache. 
Each cache content alternates from periods of inclusion and exclusion from
cache (see Figure \ref{fig:cache-in-out}), and we have (renewal theory),
\begin{equation}
\label{eq:pi_up-renew}
\pi_{up} = \frac{E[B]}{E[R] + E[B]}
\end{equation}
where $E[R]$ is the mean time that a content takes to return to the cache once it is evicted and
$E[B]$ is
the mean time that the content remains in the cache after insertion.
It should be clear that the reinforced counter mechanism
can be modeled as an M/M/1 queueing system, with state equal to the value of the
reinforced counter.
Then,
\begin{equation}
\label{eq:pi_up}
\pi_{up} = \sum_{i=K+1}^\infty (1-\rho) \rho^i = {\rho^{K+1}}
\end{equation}
where $\rho=\lambda/\mu$.   

Let $\gamma(K, \mu)$ be the rate at which content enters the cache.
Due to flow balance, in steady state, $\gamma(K, \mu)$ 
equals the rate at which content leaves the cache,
\begin{equation}
\label{enter}
\gamma(K,\mu) = \mu \rho^{K+1}(1-\rho) = \lambda \rho^K ( 1-\rho) = \frac{1} {E[B] + E[R]} 
\end{equation}
where the last equality can be easily inferred from Figure \ref{fig:cache-in-out} (renewal arguments).

$E[B]$ can be calculated from first passage time arguments, that is the time it
takes from the system to return to state $K$ (eviction) once it is brought into system
(state $(K+1)$).
From the M/M/1 model, it can also be calculated
by the busy period of an M/G/1 queue, which equals 
\begin{equation}
\label{eq:expectB}
E[B]=1/(\mu-\lambda) .  
\end{equation}
Then, from (\ref{eq:expectB}) and (\ref{eq:pi_up-renew})
\begin{equation}
\pi_{up} = \frac{1/(\mu-\lambda)}{1/(\mu-\lambda) + E[R]}  .
\end{equation}
Once the content is evicted the mean time for it to return to the cache is given by 
\begin{equation}
E[R]=(1-\pi_{up})/(\pi_{up}(\mu-\lambda)) \end{equation}

Given a fixed $\pi_{up}$, it is possible to write $\rho$,  $\mu$,  $E[R]$ and $\gamma$ as a function of $K$,
\begin{eqnarray}
\rho &=& {\pi_{up}^{1/(K+1)}} \label{eq:rho} \\
\mu &=& \lambda \pi_{up}^{-1/(K+1)} \label{eqmu}  \label{eq:mu} \\
E[R] &=& (1-\pi_{up})/(\pi_{up}(\lambda (\pi_{up})^{-1/(K+1)} -\lambda))  \label{er} \\
\gamma &=& \lambda \pi_{up}((\pi_{up})^{-1/(K + 1)} - 1) \label{gamma}
\end{eqnarray}

Note that there is a tradeoff in the choice of $K$, as increasing the value of $K$ 
reduces the rate at which content is inserted into the cache (see equation~\eqref{enter}),
which in turn reduces the steady state costs to download the content from external sources.
The larger the value of $K$, the smaller the steady state rate at which the content enters
and leaves the cache. 
But increasing the value of $K$, also increases the mean time for the content to be reinserted into the cache once the
content is evicted.
The larger the value of $K$, the longer the requesters
for a given content will have to wait in order to be able to download the content from the cache
after it is evicted (see equation~\eqref{er}). 

Let $K_{max}$ be the maximum  value allowed for  $K$. 
Motivated by the tradeoff above, given a fixed value of  $\pi_{up}$ 
we consider the
following optimization problem,
\begin{eqnarray}
\min_{K}  && \psi(K)= \alpha \gamma + \beta E[R]  \label{eq1} \\
\textrm{such that} \\
K &\leq& K_{max} \\
 \pi_{up} &=& {\rho^{K+1}} \label{eq2} \\
\rho &=& \lambda/(\lambda \pi_{up}^{-1/(K+1)}) \label{eq3} 
\end{eqnarray}
$\alpha$ and $\beta$ are used to control the relevance of long term and short term dynamics.
The long term dynamics reflect the behavior of the system after a long period of time, during
which the rate at which content enters the cache is given by $ \mu \rho^{K+1}(1-\rho)$.  
The short term dynamics reflect the behavior of the system during a shorter period of time,
during which one wants to guarantee that the mean time it takes for the content to return to the cache
is not too large.  
We should keep in mind that we choose $\pi_{up}$ 
to satisfy cache capacity limitations and system performance.

Substituting~\eqref{eq:rho}-\eqref{gamma} into ~\eqref{eq1}, the objective function is given by,
\begin{equation}
\psi(K)=\alpha \Big[ \lambda \pi_{up}((\pi_{up})^{-1/(K + 1)} - 1)  \Big] +\beta (1-\pi_{up}) \frac{1}{\pi_{up}\lambda} \Big[ \frac{1 }{ (\pi_{up})^{-1/(K+1)}-1} \Big]  \label{eqpik}
\end{equation}

The value of $E[R]$ must be bounded so as to avoid starvation, as
formalized in the propositions below.

\begin{proposition}
If $\beta=0$, the optimal strategy consists of setting $K=K_{max}$.  If $K_{max}=\infty$, it will
take infinite time for the content to be reinserted in the cache once it is evicted for the first time.
\end{proposition}
\begin{proof}[Proof:]
The objective function is given by
\begin{equation}
\psi(K)= \alpha \lambda \pi_{up} ((\pi_{up})^{-1/(K + 1)} - 1)  \label{eqpsi}
\end{equation}
The derivative of the expression above with respect to $K$ is 
\begin{equation}
\frac{d \psi(K)}{dK}=\alpha \frac{\lambda \log(\pi_{up}) \pi_{up} }{(K + 1)^2 (\pi_{up})^{1/(K + 1)}}
\end{equation}
which is readily verified to be always negative. Therefore, the minimum is reached when $K=K_{max}$.
  When $K=\infty$ it follows from \eqref{eqmu} that
$\mu$ tends to 0.  The mean time for reinsertion of the content in 
the cache is given by \eqref{er} which grows unboundedly as $\mu$ tends to 0.   
\end{proof}

\begin{proposition}
If $\beta>0$ the optimization problem ~\eqref{eq1}-\eqref{eq3} admits a unique minimum~$K^{\star}$.
\end{proposition}
The proof is omitted for conciseness.

\subsection{Illustrative Example}
\label{sec:single-t-ex}

Next, we consider an illustrative example to show the tradeoff in the choice of $K$. 
Let $\pi_{up}=0.9$ and $\alpha=\beta=1$. 
Figure~\ref{myfigil} shows how cost first decreases and then increases, as $K$ increases.
The optimal is reached for $K=10$.  At that point, we have $E[R]=0.31$ and $\gamma=0.32$.

\begin{figure}[h!]
\center
\includegraphics[scale=0.3]{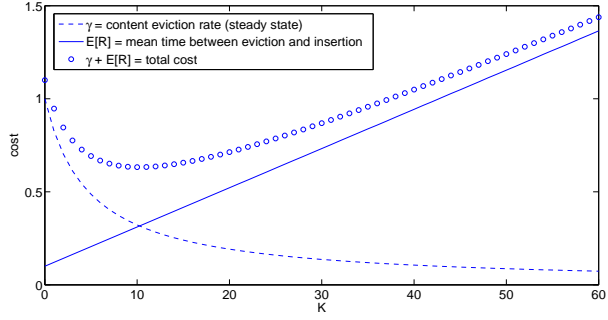}
\caption{$\pi_{up}=0.9$, $\alpha=\beta=1$} \label{myfigil}
\end{figure}

\textbf{Remark 1: }  Note that in problem~\eqref{eq1}-\eqref{eq3} we do not set a hard
constraint on $E[R]$.  Alternatively, we could add such a constraint, $E[R] \leq R^{\star}$.
 The solution of the 
modified problem is either the solution of the problem without the constraint (in case
the constraint is inactive) or the value of $K$ which yields  $E[R]=R^{\star}$
(in case the constraint is active).

\textbf{Remark 2: } If follows from Markov inequality that the solution of the 
problem~\eqref{eq1}-\eqref{eq3} naturally yields
a bound on the probability $P(R > r)$, i.e., $P(R > r ) < E[R]/r$.  
In the numerical 
example above, when $K$ is optimally set we have that the probability that the the 
content is
not reinserted into cache after $3.1$ units of time following an  eviction is  
$P(R > 0.31 \times 10) < 0.1$.  
In the numerical results we present latter $P(R > r)$ is calculated exactly
from the model.

\textbf{Remark 3: } We assume that $K$ can take real values. If $K$ is not an integer,
we can always randomize between the two closest integers when deciding whether to store or 
not the content.

\subsection{Reinforced counter with hysteresis}
\label{sec:hyst}

In this section we generalize the reinforced counter to allow for a third control knob
$K_{h}$ as follows.
The counter is incremented at each arrival request for a content and is decremented
at rate $1/\mu$.
In addition, the content is included into cache when the value of
the reinforced counter is incremented to $K+1$, as in previous section.
However, content is not removed from cache when the counter value is decremented
from $K+1$ to $K$.
Instead, content remains in cache until the counter reaches the (new) threshold $K_h$.
Figure \ref{fig:cache-in-out-hyst} illustrates the behavior of the new reinforced counter.
\begin{figure}[h!]
\center
\includegraphics[scale=0.6]{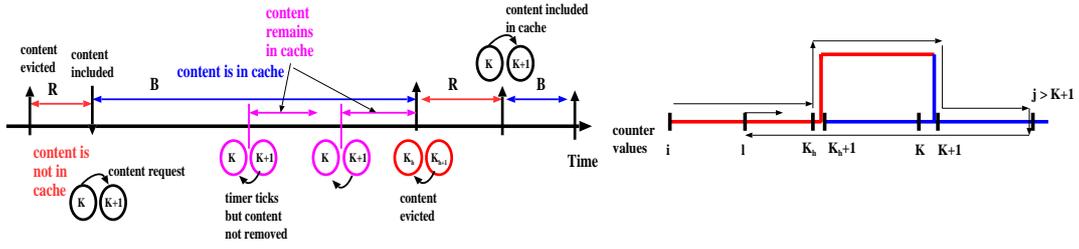}
\caption{Reinforced counter with hysteresis.
\label{fig:cache-in-out-hyst}
}
\end{figure}
From the figure, we observe that the purple intervals correspond to the values of the
reinforced counter between $K$ and $K_h$ and the content is in cache. 
The bottom of Figure \ref{fig:cache-in-out-hyst} shows a trajectory of the
reinforced counter that motivates the name \textit{hysteresis}.
The displayed trajectory is
$i \rightarrow K_h \rightarrow K_h+1 \rightarrow K \rightarrow K+1 \rightarrow 
j \rightarrow K_h+1 \rightarrow l$ etc.

At first glance, it seems that only the intervals where the 
content is in cache are affected by this new mechanism (the blue intervals)
but not the intervals where content is absent (the red intervals).
However, this is not true and both intervals are affected.
In what follows, we show how to calculate the expected values of $E[B]$ and $E[R]$
and how the measures of interest are affected by this new mechanism.
We also show the advantages of the reinforced counter with hysteresis.
For that, we refer to Figure \ref{fig:cache-in-out-hyst-MC}
that shows the Markov chain for the hysteresis counter.
\begin{figure}[h!]
\center
\includegraphics[scale=0.6]{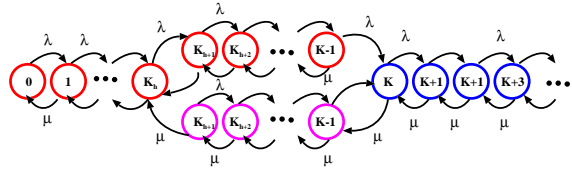}
\caption{Reinforced counter with hysteresis: state transition diagram.
\label{fig:cache-in-out-hyst-MC}
}
\end{figure}

Let $\nu(i+1) = E[B]$ and
$\xi(i+1) = E[R]$ both when $K - K_h = i$.
Note that $\nu(1)$ and $\xi(1)$ are the values of $E[B]$ and $E[R]$ for a cache
with no hysteresis.
$\nu(1)$ is obtained from equations \eqref{eq:pi_up-renew}, \eqref{eq:pi_up} and \eqref{er}.
$\xi(1) = \left({\sum_{j=0}^K 1/\rho^{j} }\right) / \left({ \lambda/(\lambda+\mu) }\right) $.

\begin{proposition}
\label{prop:nu_xi}
$\nu(i+1)$ and $\xi(i+1)$ can be obtained by the following recursions:
\begin{eqnarray}
\nu(i) &= \frac{1}{\mu} + \nu(i-1) + \rho \nu(1)          \;\; &i \geq 2 \label{eq:nu-rec} \\ 
\xi(i) &= \frac{1}{\lambda} + \xi(i-1) + \frac{1}{\rho} \xi(1) &2 \leq i \leq K-K_h+1  \label{eq:xi-rec} 
\end{eqnarray}

\end{proposition}
\begin{proof}[Proof:]
The proof follows from renewal arguments, and is omitted due to space limitations.
\end{proof}
It is important to note that explicit expressions for $\nu(i)$ and $\xi(i)$ can 
be obtained as a function of $\lambda$, $\mu$ and $K$ but details are omitted since the
recursion above suffices to explain our comments.

Suppose $\lambda$ and $\pi_{up}$ are given and we obtain
$K$ and $\mu$, for instance from the optimization problem in the previous section. 
We allow $K_h$ to vary from $K_h=K$ (that is reinforced counter without hysteresis) to $K_h=0$.
\begin{proposition}
\label{prop:gamma}
As $K_h$ decreases, the rate $\gamma(K,\mu)$ at which content enters the cache also decreases.
\end{proposition}
\begin{proof}[Proof:]
From Proposition \ref{prop:nu_xi}, it is not difficult to see that
both $E[B]$ and $E[R]$ increase with $K_h$.
Since $\gamma(K,\mu) = 1/(E[B] + E[R])$ (equation \eqref{enter}) the result follows.
\end{proof}

Proposition \ref{prop:gamma} shows that, from an initial value of $\pi_{up}$, if we fix the 
parameters $\lambda$ and $K$,
the rate at which content is replaced (both included and removed from cache) decreases
by using the hysteresis mechanism, which is good to lower costs as explained in the
previous section.
However, $\pi_{up}$ also varies.
As a consequence of Proposition \ref{prop:nu_xi} we can show that $\pi_{up}$ is reduced.
This is not obvious since $E[B]$ increases.
But, by adjusting the knob $\mu$, $\pi_{up}$ can be maintained constant while $\gamma$ is reduced
when the hysteresis schema is used.
The proof of this last result is omitted but it follows from Proposition \ref{prop:nu_xi}.
The numerical results presented in the Section \ref{sec:numerical} corroborate the claim.

There are additional advantages of the hysteresis mechanism.
First note that, in the previous sections, we assumed that file download times are negligible. 
However, when a user requests for a content it is important that the whole file remains
stored at the cache not only until this user finishes downloading
but also while other users are downloading the same content from that cache. 
Hysteresis is helpful to prevent the file from being
removed before its download is concluded by all requesters.
In Section \ref{sec:numerical} we show that hysteresis increases the probability that a content
remains in cache for at least some time $t$ after it is cached.
Note that this is an additional (transient) performance measure and it differs from the $\gamma$ metric
(steady state rate).
As our numerical results show, by adjusting $K_h$ we can improve both steady state and transient metrics.

Our numerical results also show that hysteresis reduces the
coefficient of variation of the time that content resides in the cache ($B$). 
This is important for cache capacity planning with multiple contents, 
since more predictable systems are usually easier to design and control.
%

\subsection{Numerical Results}
\label{sec:numerical}

In this section we show some numerical results obtained for the models: 
the reinforced counter 
with a single threshold~($K$) and the reinforced counter with two thresholds~($K$,$K_h$). 
Three performance measures were used to analyze the models: 
the rate at which content enters the cache ($\gamma$), 
the cumulative 
distribution of the time the content takes to return to the cache~(R)
and the cumulative 
distribution of the time the content remains in the cache after insertion~(B).

Three scenarios were evaluated: (a) the reinforced counter has a 
single threshold $K=11$, (b) the reinforced counter has two thresholds
$K=11$ and $K_h=7$, and (c) the reinforced counter has two thresholds
$K=11$ and $K_h=2$. 
For each scenario, we consider the same value of $\pi_{up}$, $\lambda$,
and $K$. 

The value of $\gamma$ for scenario (a) is $0.24$, for (b) is $0.07$ 
and for (c) is $0.04$. 
We note that the rate at which content enters or leaves the cache decreases
as the value of $K_h$ decreases. 
This is one of the advantages of introducing a third control knob
$K_{h}$.  

Figures \ref{fig:dist-R-B}(a) and \ref{fig:dist-R-B}(b) show the cumulative 
distribution of R and B. 
We note that the reinforced counter with hysteresis allow to control the
probability distribution of R and B. 
In Figure \ref{fig:dist-R-B}(a), consider for example $t=4$.
If we set $K_h=2$, we have $P[R<4]=0.35$ and if $K_h=7$, then $P[R<4]=0.65$.
As the value of $K_h$ increases, the probability of R be less than a
certain value of $t$ increases.
This behavior can also be observed  for the distribution of B. 
On the other hand, if we consider the model with a single threshold
we can not control the distribution of R and B.
In Figure \ref{fig:dist-R-B}(a), $P[R<4]=0.9$. 

Another advantage of the reinforced counter with hysteresis, is that
the coefficient of variation of R and B, decreases with the value of
$K_h$, which means that the dispersion of the distribution of R and B also
decreases. 
The values obtained for the coefficient of variation of B for each scenario are: 
(a) $1.6$, (b) $1.3$ and (c) $1.1$. 

\begin{figure}[htb]
\hspace{-0.3in}
\includegraphics[scale=0.4]{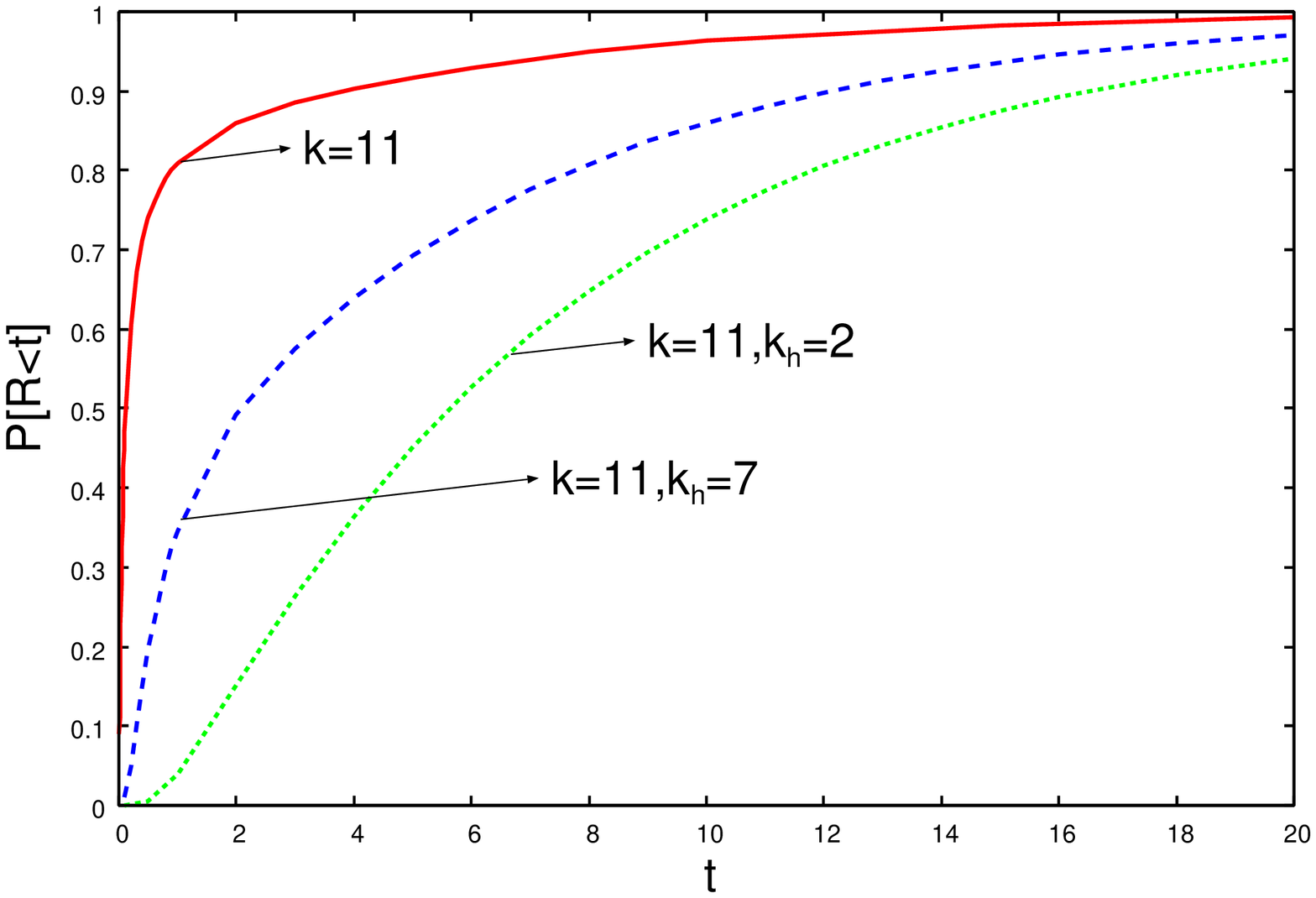}
\includegraphics[scale=0.4]{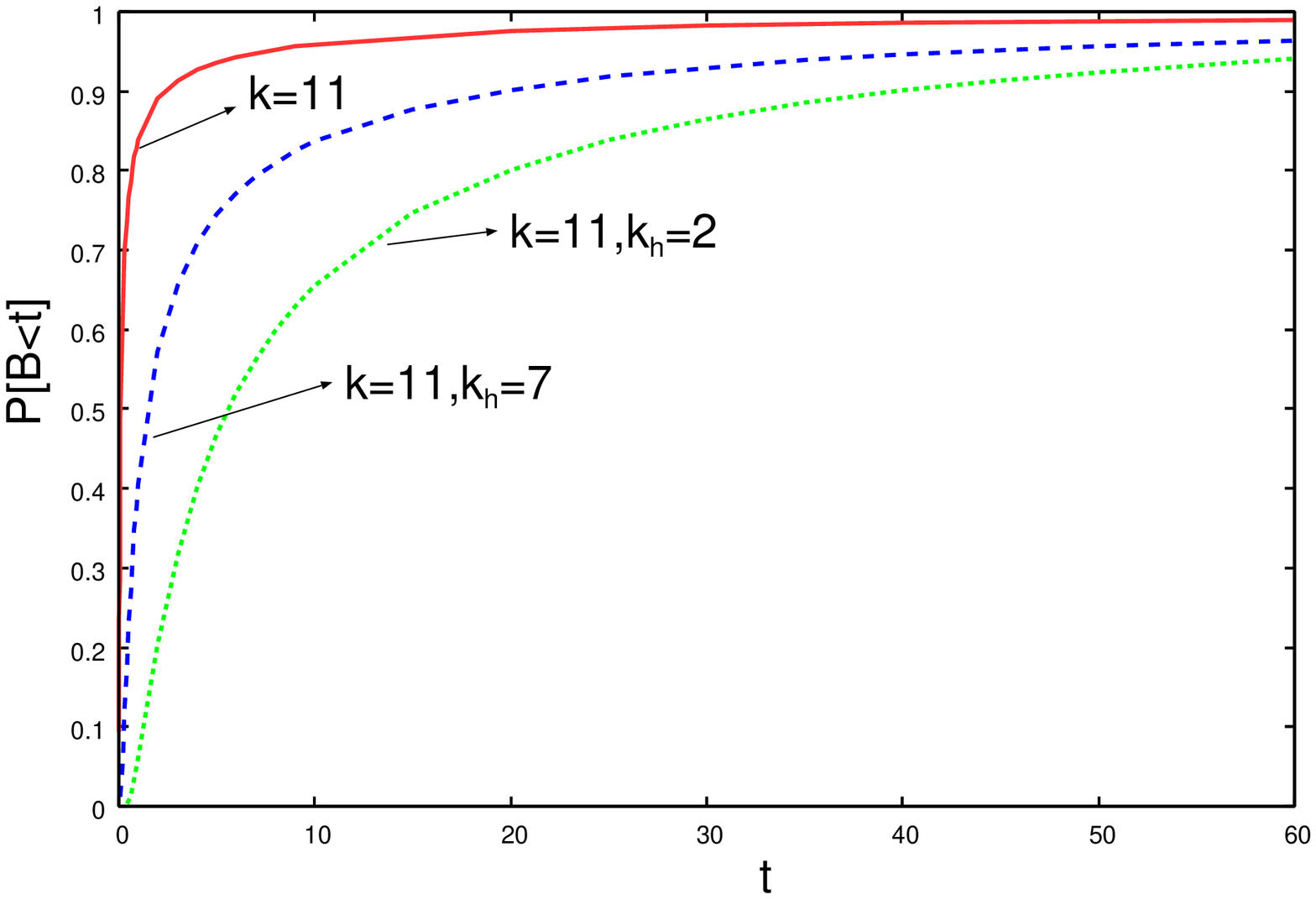}
\center
\begin{tabular*}{1.0\textwidth}{@{\extracolsep{\fill} }ll}
(a) Cumulative Distribution of R & (b) Cumulative Distribution of B
\end{tabular*}
\caption{Cumulative Distribution of the time the content takes to return/remains 
in the cache.} 
\label{fig:dist-R-B}
\end{figure}

\section{Multiple Caches} 
\label{model}

In the previous section we argued that we can decouple the cache storage dynamics of each content
and study each content in isolation.
Roughly, decoupling is a consequence of statistical multiplexing, for content placement
policies that limit the amount of time a single content is cached, independently of
the other requested contents.
This is the case of the \textit{reinforced counter} policies we study.
We showed the advantages of these policies and the flexibility they bring to the design of caches.

In this section we address the problem of multiple caches in a network.
The main objectives are:
(a) to formulate the problem of content placement at a cache network; 
(b) to show that if we do not use a decoupling policy such as those we study in this paper,
the optimal content placement problem is NP hard.  
This last result emphasize the importance of employing a
cache placement mechanism like the reinforced counters.

\subsection{Model and Problem Formulation}

Let $\mathcal{F}$ and $\mathcal{C}$ be the set of files and caches in the system. 
There are $F=\mathcal{|F|}$ files and $C=\mathcal{|C|}$ caches in the system.   
File $f$ has size $t_f$, $1 \leq f \leq F$ and cache $i$ has (storage) capacity $s_i$, $1 \leq i \leq C$.  
Users issue exogenous requests for files. At each cache $i$, $1 \leq i \leq C$, 
exogenous requests for file $j$ arrive at rate $\lambda_{ij}$.  
We assume that cache $i$ has service capacity $\eta_i$ requests/s.  
Once a request arrives at a cache, 
\begin{enumerate}
\item \textbf{in case of  a cache hit}, the content is immediately transferred back to the requester 
through the data plane.  Our model is easily
adapted to account for the delay in searching for a  file in the cache, but for the sake of
presentation conciseness we assume zero searching cost in this paper;
\item \textbf{in case of a cache miss}, the request is added to the queue of requests to be serviced. 
When the request reaches the head of the queue,
it is transferred to one of the outgoing links, selected uniformly at random.
\end{enumerate}
\begin{figure}
\center
\includegraphics[scale=0.7]{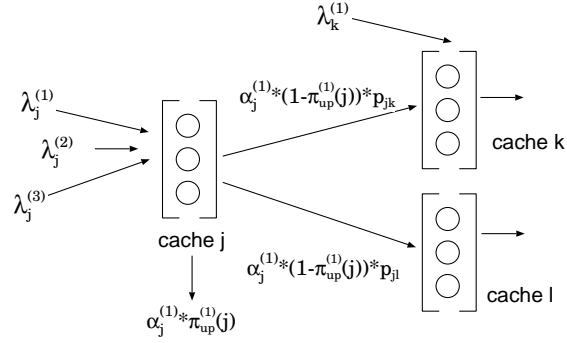}
\caption{Illustration of cache network}
\label{fig:cachenet}
\end{figure}
%

Let $A_{ij}$ be a variable that indicates if file $j$ is available at cache $i$.
>From the results of previous section, $A_{ij}=1$ with probability $\pi_{up}$ for content $j$
at cache $i$.
Therefore, we can compute $P[A_{ij}=1]$ using the derived results.
In addition, we may also study cases
in which content is statically placed into caches, i.e., the content in 
the caches is not replaced. 
This corresponds, for instance, to a system where the replacements of files from the caches occur at
a much coarser granularity than the requests.
In this case, $A_{ij}=1$  if file $j$ is at cache $i$ and 0 otherwise.

Let $M_{hi}$ be the adjacency matrix between caches, i.e., $M_{hi}=1$
if there is a direct path from cache $h$ to cache $i$ and 
0 otherwise.

Let $d_h$ be the outgoing degree of node $h$. 
Let $p_{hi}$ be the probability that a request from cache $h$ is routed to cache $i$. 
In this paper, we consider random walk routing (\cite{sbrc13,wperf14}), i.e., 
$p_{hi}= 1/|d_h|$ if $M_{hi} > 0$ and 0 otherwise.

Figure~\ref{fig:cachenet} illustrates the problem in the case where we have a cache network comprised of 3 nodes, $j$, $k$ and $l$, and 3 contents, indexed by 1, 2 and 3.
In the Figure, the exogenous arrival rates $\lambda_{j}^{(1)}, \lambda_{j}^{(2)}, \ldots$
of contents 1, 2 and 3 to cache $j$ are made explicit.  
Let $\pi_{up}^{(c)}(i)$ the designed parameter $\pi_{up}$ in previous sections, but for
cache $i$ and content $(c)$.
Let $\alpha_{j}^{(c)}$ represents the total input rate of requests to content $(c)$ arriving at cache $j$
(exogenous arrivals to cache $j$ plus those requests coming from other caches).
A fraction $\alpha_{j}^{(c)} \pi_{up}^{(c)}(j)$
is immediately served, and 
fractions  $\alpha_{j}^{(c)} (1 - \pi_{up}^{(c)}(j) )  p_{jk}$ 
and $\alpha_{j}^{(c)} (1 - \pi_{up}^{(c)}(j) ) p_{jl}$
are transferred to caches $k$ and $l$,
respectively.  

%
%

In the remainder of this paper, we will assume that requests for files that are cached at $i$, 
$1 \leq i \leq C$, are immediately processed at $i$, and we will
be concerned  with  transfers that occur due to misses.  
\begin{equation}
\alpha_{j}^{(c)} = 
\lambda_{j}^{(c)}   + \sum_{k=1,k \neq j}^C \alpha_{k}^{(c)} (1-\pi_{up}^{(c)}(k)) p_{kj}
\end{equation}
%
%

We can compute the mean response time from the  expected total number of requests in the system.
Using Little's law, and recalling that the processing rate is $\eta_i$ we have that the average number 
of requests at cache $i$ is $E[N_i]=\alpha_i /\eta_i$.
Then, the expected total number of requests in the system is $E[N]=\sum_{i=1}^C E[N_i]$ and the mean response 
time, $E[T]=E[N]/\Lambda$ where $\Lambda$ is the sum of the exogenous arrival rates to all caches.

\begin{definition}
The cache network is stable if and only if $\eta_i > \alpha_i$ for all $i$, $1 \leq i \leq C$.
\end{definition}


Until this  point, we  focused on the  use of reinforced  counters for
content placement. In the remainder of the paper, we show that without
reinforced counters  a natural statement of  optimal content placement
yields an NP hard problem.

\begin{definition} \label{defini}
The optimal content placement problem consists of finding the mapping  $A: \mathcal{F} \rightarrow \mathcal{C}$ such that
\begin{enumerate}
\item the cache network is stable
\item $\sum_{f=1}^F 1_{A_{if}=1} t_f \leq s_i$, for $i=1,\ldots F$
\item the arrival rates of requests that require processing at the control plane, $\sum_{i=1}^{C} \alpha_i$, is minimized
\end{enumerate}
\end{definition}

\subsection{Static Content Placement in Cache Networks is NP Hard}
\label{complexity}

We show that the problem of static optimal content placement in cache networks is NP hard. 
This result is interesting because it shows the importance of using a cache replacement policy
like those we study in this paper.
Recall that, by using the reinforcement policy, we can optimize the counter parameters
to satisfy a given probability of finding a content in cache for each content in isolation.
These probabilities are in turn obtained to satisfy a given cache capacity constraint,
from statistical multiplexing arguments. (the problem is identical to sizing a telephone network
and one can use the results of that area.) 

First, we consider the case where 
different files have different sizes. In this case, it is possible to show that  the problem is NP hard even if we need to make placement 
decisions at a single cache. To this aim, we can consider a mapping from {\sc Knapsack} (the proof is omitted due to space constraints). Then, we consider the simpler case wherein all files have the same size.   In this case, if we need to make 
placement decisions at a single cache, the problem can be solved using a greedy strategy. Nevertheless, we show that the problem is still
NP hard if we need to make placement decisions in at least two caches.

\begin{theorem}The optimal content placement problem is NP hard.
\end{theorem}
\begin{proof}[Proof:]
We refer to {\sc FeasibleCache} as the problem of deciding if a given cache network admits a feasible content placement.
Note that the problem of finding the optimal placement, i.e. the {\sc PlacementProblem}, must be harder than {\sc  FeasibleCache}, 
since solving the optimization problem yields  a solution to {\sc  FeasibleCache}.  We proceed with a Turing reduction from the {\sc Partition} problem to the {\sc FeasibleCache} problem.  
The complete proof is omitted due to space limitations.  \eat{ is included in the appendix for reference and } \end{proof}

\section{Related Work} \label{related}

The literature on caching~\cite{caching} and content placement~\cite{lopresti1}, and its relations to networking~\cite{neves2014new}, database systems  and operating systems~\cite{nelson1988caching} is vast.  Nevertheless,
the study of cache networks, which encompass networks where routing and caching decisions are taken together
at devices which work as routers and caches is scarce~\cite{elisha1}.  In this work, we present a systematic analysis of cache networks using reinforced counters and indicating their applicability both in the single cache as well as multiple cache settings. 

Cache networks play a key role in content centric networks (CCNs)~\cite{ccn1, ccn2}.   CCNs are emerging
as one of the potential architectures for the future Internet.  In CCNs, content rather than hosts is addressed.  
The \emph{glue}  that binds CCNs are the content chunks rather than the IP packets~\cite{ccn1}.

\eat{
Two classical and distinct solutions to the content placement problem consist of 1) placing a fraction of every content in each cache 
 ~\cite{towsley}  or 
2) placing in caches only the most popular contents~\cite{fayazbakhsh2013less}.  In this work, we 
assume  the use of reinforced counters for content placement.  When considering the hardness of content placement, we assume that every content must be placed in at least one  
cache.   Whereas previous work assumes that content can always be retrieved from a custodian in case it is not available in the caches,
in this paper we assume that the caches are the only sources of content.  Previous works assumed that caches were used exclusively 
to increase performance, while  in this work we consider publishers that leverage caches to reduce service costs.  }

Using caches to reduce publishing costs can be helpful for publishers that cannot afford staying online all the time, or need to
 limit the bandwidth consumed to replicate content.  With caches, content can \emph{persist} in the network even in the absence of a publisher~\cite{ccn2}.
 In case of intermittent publishers, content availability can only be guaranteed in case replicas of the content are \eat{permanently } placed
 in some of the caches. \eat{, as considered in this paper.  }
 
 Approximate and exact analysis of single caches and cache networks has been studied for decades~\cite{exact, bianchi2013general, towsley1}. Approximate analysis is usually carried out using mean field approaches~\cite{bianchi2013general} or other asymptotic techniques~\cite{jelenkovic1999asymptotic}, whereas exact analysis is performed accounting for  Markov chain properties~\cite{fagin1978efficient} or assuming uncoupled content dynamics~\cite{towsley1}.  In this paper, we consider reinforced counters as the placement algorithms, which are at the same time amenable to analytical study and of practical interest in the context of DNS systems or the novel Amazon ElastiCache~\cite{amazon, exact}.  Cache policies are traditionally devised to avoid staleness of content and/or improve system efficiency.  The analysis presented in this paper shows that reinforced counters can be used to target these two goals.

\section{Conclusion and Future Work} 
\label{concl}

In this  paper we have studied  the problem of  scalable and efficient
content placement in cache  networks.  We analyzed an opportunistic caching policy
using reinforced counters
(\cite{sbrc13}) as an efficient mechanism for content placement and proposed an
extension to the basic scheme (reinforced counters with hysteresis). 
Using reinforced counters, we indicated how to tune simple knobs in order to
reach the optimal  placement.  We then showed  some of the
properties of the mechanism and  the optimal content placement  problem. 
Finally, using
the  queueing-network model we propose for a  cache network,  we showed
how to calculate the overall expected delay for a request to obtain a content.
In addition, we showed that
obtaining the optimal  solution without reinforced counters is an
NP hard problem.

This work is  a first step towards the  characterization of efficient,
scalable  and tractable  content placement  in cache  networks.
Future work  consists of  studying how to  distributedly
tune the caches,  as well as allowing the  presence of custodians that
store permanent copies of the files.

\bibliographystyle{sbc}
\bibliography{cache}


\end{document}